\documentclass[11pt]{amsart} 
\usepackage{amscd,amssymb,amsxtra}
\usepackage{mathrsfs}  
\DeclareSymbolFontAlphabet{\mathrsfs}{rsfs}
\usepackage[mathscr]{eucal} 
\usepackage{comment}
\usepackage{color}
\usepackage{bbm}
\usepackage{enumitem} 
\usepackage[colorlinks,backref=page,citecolor=refkey]{hyperref}
\usepackage{aliascnt}
\usepackage{spectralsequences}

\usepackage{marginnote}

\makeatletter
\let\@secnumfont\bfseries
\def\section{\@startsection{section}{1}%
  \z@{4\linespacing\@plus\linespacing}{\linespacing}%
  {\bfseries\centering}}
\def\introsection{\@startsection{section}{1}%
  \z@{3\linespacing\@plus\linespacing}{\linespacing}%
  {\bfseries\centering}}
\def\subsection{\@startsection{subsection}{2}%
   \z@{1.25\linespacing\@plus.7\linespacing}{.5\linespacing}%
   {\normalfont\bfseries}}
\def\subsectionsinline{\def\subsection{\@startsection{subsection}{2}%
  \z@{1\linespacing\@plus.7\linespacing}{-.5em}%
  {\normalfont\bfseries}}}

\makeatother

\numberwithin{equation}{section}

\newcommand{\mynewtheorem}[2]{
  \newaliascnt{#1}{equation}
  \newtheorem{#1}[#1]{#2}
  \aliascntresetthe{#1}
  \expandafter\def\csname #1autorefname\endcsname{#2}
}

\theoremstyle{definition}
\mynewtheorem{definition}{Definition}
\mynewtheorem{example}{Example}
\mynewtheorem{problem}{\color{blue}Problem}
\mynewtheorem{probsec}{\color{blue}Problem}
\mynewtheorem{exercise}{Exercise}
\mynewtheorem{question}{\color{blue}Question}
\mynewtheorem{project}{\color{blue}Project}
\mynewtheorem{construction}{Construction}
\mynewtheorem{task}{\color{blue}Task}
\mynewtheorem{otask}{\color{green}Obsolete Task}
\mynewtheorem{notation}{Notation}
\mynewtheorem{learn}{\color{blue}Learn}

\newtheorem*{definition*}{Definition}
\newtheorem*{example*}{Example}
\newtheorem*{problem*}{\color{blue}Problem}
\newtheorem*{probsec*}{\color{blue}Problem}
\newtheorem*{exercise*}{Exercise}
\newtheorem*{question*}{\color{blue}Question}
\newtheorem*{learn*}{\color{blue}Learn}
\newtheorem*{project*}{\color{blue}Project}
\newtheorem*{construction*}{Construction}
\newtheorem*{notation*}{Notation}

\theoremstyle{remark}
\mynewtheorem{note}{Note}
\mynewtheorem{remark}{Remark}
\mynewtheorem{data}{Data}

\newtheorem*{note*}{Note}
\newtheorem*{remark*}{Remark}
\newtheorem*{data*}{Data}

\theoremstyle{plain}
\mynewtheorem{theorem}{Theorem}
\mynewtheorem{corollary}{Corollary}
\mynewtheorem{lemma}{Lemma}
\mynewtheorem{proposition}{Proposition}
\mynewtheorem{conjecture}{Conjecture}
\mynewtheorem{claim}{Claim}
\mynewtheorem{proposal}{Proposal}
\mynewtheorem{conclusion}{Conclusion}
\mynewtheorem{hypothesis}{Hypothesis}
\mynewtheorem{assumption}{Assumption}

\newtheorem*{theorem*}{Theorem}
\newtheorem*{corollary*}{Corollary}
\newtheorem*{lemma*}{Lemma}
\newtheorem*{proposition*}{Proposition}
\newtheorem*{conjecture*}{Conjecture}
\newtheorem*{claim*}{Claim}
\newtheorem*{proposal*}{Proposal}
\newtheorem*{conclusion*}{Conclusion}
\newtheorem*{hypothesis*}{Hypothesis}
\newtheorem*{assumption*}{Assumption}

\newenvironment{proof*}[1][\proofname]{
  \begin{proof}[#1]}{  
\end{proof}}

\definecolor{refkey}{rgb}{0,.6,.4}

\renewcommand{\:}{\colon}

\newcommand{\CC}{{\mathbb C}}

\DeclareMathOperator{\Conf}{Conf}

\DeclareMathOperator{\End}{End}

\newcommand{\QQ}{{\mathbb Q}}

\newcommand{\RR}{{\mathbb R}}
\newcommand{\TT}{\mathbb T}
\DeclareMathOperator{\Spin}{Spin}

\DeclareMathOperator{\tr}{tr}

\newcommand{\ZZ}{{\mathbb Z}}

\newcommand{\chiup}{\raise.5ex\hbox{$\chi$}}

\newcommand{\dbar}{{\overline\partial}}

\DeclareRobustCommand{\mstrut}{^{\vphantom{1*\prime y\vee M}}}

\newcommand{\temsquare}{\raise3.5pt\hbox{\boxed{ }}}

\newcommand{\zmod}[1]{\ZZ/#1\ZZ}

\newcommand{\zt}{\zmod2}

\DeclareFontFamily{U}{mathx}{}
\DeclareFontShape{U}{mathx}{m}{n}{<-> mathx10}{}
\DeclareSymbolFont{mathx}{U}{mathx}{m}{n}
\DeclareMathAccent{\widehat}{0}{mathx}{"70}
\DeclareMathAccent{\widecheck}{0}{mathx}{"71}
 
\DeclareMathSymbol{\bigtimes}{1}{mathx}{"91}

\DeclareMathOperator{\SO}{SO}

\let\O\relax
\DeclareMathOperator{\O}{O}
\DeclareMathOperator{\U}{U}
\DeclareMathOperator{\SU}{SU}
\DeclareMathOperator{\GL}{GL}

\newcommand{\upplus}{^{>0}}
\newcommand{\upnn}{^{\ge0}} 
\newcommand{\Zp}{\ZZ\upplus}
\newcommand{\Znn}{\ZZ\upnn}
\newcommand{\Rp}{\RR\upplus}

\usepackage[all,2cell]{xy}
\usepackage{graphicx}\usepackage{epsf}

\usepackage{tikz-cd}

\definecolor{refkey}{rgb}{0,.8,.2}\definecolor{labelkey}{rgb}{1,0,0} 

\newcommand{\bmuu}{\mbox{$\raisebox{-.07em}{\rotatebox{9.9}
  {\tiny {\bf /}
  }}\hspace{-0.49em}\mu\hspace{-0.88em}\raisebox{-0.98ex}{\scalebox{2} 
  {$\color{white}\phantom{.}$}}\hspace{-0.416em}\raisebox{+0.88ex}
  {$\color{white}\phantom{.}$}\hspace{0.46em}$}} 
\newcommand{\bmut}{\bmu 2}
\newcommand{\bmu}[1]{\bmuu _{#1}}

\DeclareMathOperator{\Bord}{Bord}
\DeclareMathOperator{\Euler}{Euler}
\DeclareMathOperator{\Line}{Line}
\DeclareMathOperator{\Man}{Man}
\DeclareMathOperator{\Orient}{Orient}
\DeclareMathOperator{\Pfaff}{Pfaff}
\DeclareMathOperator{\Riem}{Riem}
\DeclareMathOperator{\Set}{Set}
\DeclareMathOperator{\Vect}{Vect}
\DeclareMathOperator{\fr}{fr}
\DeclareMathOperator{\pfaff}{pfaff}
\newcommand{\BGL}{B\!\GL}
\newcommand{\BN}{B\mstrut _{\nabla }}
\newcommand{\BO}{B\!\O}
\newcommand{\BSO}{B\!\SO}
\newcommand{\BSpin}{B\!\Spin}
\newcommand{\Btt}{\Bord_{\langle 2,3  \rangle}}
\newcommand{\Fl}{F_{\lambda }}
\newcommand{\ICx}{I\CC^\times }
\newcommand{\Kh}{K^{1/2}}
\newcommand{\MTSO}{MT\!\SO}
\newcommand{\MTSpin}{MT\!\Spin}
\newcommand{\Mant}{\Man_3}
\newcommand{\RZ}{\RR/\ZZ}
\newcommand{\SZ}[1]{\Sigma ^{#1}\ZZ}
\newcommand{\TLC}{\Theta ^{\textnormal{LC}}}
\newcommand{\YC}{YM$+$CS}
\newcommand{\Zl}{Z_{\lambda }}
\newcommand{\al}{\alpha \mstrut _\lambda }
\newcommand{\bFl}{\overline{F}_\lambda }
\newcommand{\bone}{\mathbbm{1}}
\newcommand{\fp}{\mathfrak{p}}
\newcommand{\gmh}{\gamma _{-1/2}}
\newcommand{\hC}{\widehat{C}}
\newcommand{\hF}{\widehat{F}}
\newcommand{\hM}{\widehat{M}}
\newcommand{\hZ}{\widehat{Z}}
\newcommand{\ha}{\widehat{\alpha }}
\newcommand{\hc}{\hat{c}}
\newcommand{\hta}{\widehat{\tau }}
\newcommand{\op}{^\textnormal{op}}

\newcommand{\sE}{\mathscr{E}}
\newcommand{\sFBN}[1]{\sF\!_{\BN#1}}
\newcommand{\sFCwwp}{\sF_{\Conf,w_1,w_2,p_1}}
\newcommand{\sFCww}{\sF_{\Conf,w_1,w_2}}
\newcommand{\sFRwp}{\sF_{\Riem,w_1,p_1}}
\newcommand{\sFRwwp}{\sF_{\Riem,w_1,w_2,p_1}}
\newcommand{\sFRww}{\sF_{\Riem,w_1,w_2}}
\newcommand{\sFRw}{\sF_{\Riem,w_1}}
\newcommand{\sFwp}{\sF_{w_1,p_1}}
\newcommand{\sFwwp}{\sF_{w_1,w_2,p_1}}
\newcommand{\sFww}{\sF_{w_1,w_2}}
\newcommand{\sFw}{\sF_{w_1}}
\newcommand{\sF}{\mathscr{F}}
\newcommand{\sP}{\mathcal{P}}
\newcommand{\sX}{\mathscr{X}}
\newcommand{\sY}{\mathscr{Y}}
\newcommand{\tV}{t\!\Vect}
\newcommand{\tal}{\tilde\alpha \mstrut _\lambda }
\newcommand{\wpo}{(w_1,\,p_1)}
\newcommand{\wwpo}{(w_1,\,w_2, \,p_1)}
\renewcommand{\SS}{\mathbb{S}}

\begin{document}

\abovedisplayskip18pt plus4.5pt minus9pt
\belowdisplayskip \abovedisplayskip
\abovedisplayshortskip0pt plus4.5pt
\belowdisplayshortskip10.5pt plus4.5pt minus6pt
\baselineskip=15 truept
\marginparwidth=55pt

\makeatletter
\renewcommand{\tocsection}[3]{%
  \indentlabel{\@ifempty{#2}{\hskip1.5em}{\ignorespaces#1 #2.\;\;}}#3}
\renewcommand{\tocsubsection}[3]{%
  \indentlabel{\@ifempty{#2}{\hskip 2.5em}{\hskip 2.5em\ignorespaces#1%
    #2.\;\;}}#3} 
\renewcommand{\tocsubsubsection}[3]{%
  \indentlabel{\@ifempty{#2}{\hskip 5.5em}{\hskip 5.5em\ignorespaces#1%
    #2.\;\;}}#3} 
\def\@makefnmark{%
  \leavevmode
  \raise.9ex\hbox{\fontsize\sf@size\z@\normalfont\tiny\@thefnmark}} 
\def\multfoot{\textsuperscript{\tiny\color{red},}}
\def\footref#1{$\textsuperscript{\tiny\ref{#1}}$}
\makeatother

\setcounter{tocdepth}{2}


 \title[3d Chern--Simons theories and $p_1$-structures]{The role of
$p_1$-structures in \\ 3-dimensional Chern--Simons theories}

 \author[D. S. Freed]{Daniel S.~Freed}
 \thanks{DSF is supported by the Simons Foundation Award 888988 as part of the
Simons Collaboration on Global Categorical Symmetries.  This work was performed
in part at Aspen Center for Physics, which is supported by National Science
Foundation grant PHY-2210452.}
 \address{Harvard University \\ Department of Mathematics \\ Science Center
Room 325 \\ 1 Oxford Street \\ Cambridge, MA 02138}
 \email{dafr@math.harvard.edu}

  \author[C. Teleman]{Constantin Teleman} 
  \thanks{CT is supported by the Simons Foundation Award 824143 as part of the
Simons Collaboration on Global Categorical Symmetries.}
  \address{Department of Mathematics \\ University of California \\ 970 Evans
 Hall \#3840 \\ Berkeley, CA 94720-3840}  
  \email{teleman@math.berkeley.edu}

 \dedicatory{To S.-T. Yau, stronger than ever at 75}
 \date{May 25, 2026}
 \begin{abstract} 
 Our recent paper~\cite{FST} with Claudia Scheimbauer uses the cobordism
hypothesis to construct fully local Chern--Simons theories.  Here we expose
some physics motivations: Yang--Mills plus Chern--Simons in the bosonic case
and the free Majorana--Weyl spinor field in the fermionic case.  We also give
expositions of tangential structures and invertible field theories, in
particular the `gravitational Chern--Simons theory' used by Witten to obtain
topological field theories from the underlying gauge theory.
 \end{abstract}
\maketitle

In the mid-late 1980s new invariants of links and 3-manifolds were introduced
in rapid succession, first by Jones~\cite{Jo1,Jo2} and others~\cite{FYHLMO}
using von Neumann algebras and more traditional methods, by Witten~\cite{W1}
using quantum field theory, and by Reshetikhin--Turaev~\cite{RT1,RT2} using
quantum groups.  Witten's approach inspired decades of mathematical research in
topological field theory that continues unabated.  (The survey~\cite{ABHH}
recounts the first 20~years of developments.)  Recently, together with Claudia
Scheimbauer~\cite{FST}, we used the cobordism hypothesis~\cite{Lu} to construct
fully local topological field theories that encompass the invariants introduced
by Witten--Reshetikhin--Turaev.  In this expository paper we revisit a few
field-theoretic aspects of this story.
 
We begin and end with nontopological quantum field theories that surround these
invariants.  In~\S\ref{sec:1} we consider 3-dimensional Yang--Mills theory with
a Chern--Simons term~\cite{Sch,DJT} and take a singular limit in which all
massive modes disappear.  While not precisely Witten's approach
in~\cite{W1}---there he studies pure Chern--Simons theory, though with some
regularization---this limit is expected to produce the same theory~\cite{W3}.
More precisely, the assertion that the limit is topological is only true for
the underlying projective theories: a slight metric dependence remains in the
limiting linear theory.  Witten trades this slight metric dependence for a
dependence on a tangential structure for which there are various choices: a
framing~\cite{W1}, a ``2-framing''~\cite{A}, or a $p_1$-structure~\cite{BHMV}.
(In~\cite[\S7]{FST} we introduce another possibility: a complex
$p_1$-structure.)  We explain this maneuver in terms of invertible field
theories.  In particular, we tensor the singular limit with a classical
Chern--Simons theory~$\gamma _{-c}$ that is a secondary invariant of a multiple
of the first Pontrjagin class, and for that reason $p_1$-structures are most
natural here.\footnote{In the context~\cite{FST} of the cobordism hypothesis,
framings are most natural, at least at first.}  The multiple is determined by
the \emph{central charge}~$c$, a rational number computed~\eqref{eq:59} from
the level of the Chern--Simons term.  The topological field theory (which in
topology is known as \emph{Chern--Simons theory}) obtained after tensoring
with~$\gamma _{-c}$ only ``knows'' the central charge modulo~24.
 
An important aspect of the physics is the 2-dimensional chiral conformal
Wess--Zumino--Witten theory~\cite{WZ,W4} that lives on the boundary of
topological Chern--Simons theory.  We do not discuss it here, but
in~\S\ref{sec:6} we take up another chiral 2-dimensional quantum field theory:
the free Majorana--Weyl spinor field.  We give three Wick-rotated variants.  In
the last we use an analog of the Witten maneuver to express it as the boundary
theory of a topological field theory, the topological field theory whose
partition function is the Adams $e$-invariant.  In fact, the Witten maneuver
here is an extension to invertible field theories of the Atiyah--Patodi--Singer
expression~\cite[\S4]{APS} for the Adams $e$-invariant.
 
The exposition highlights aspects of Wick-rotated field theory that require
attention and some further development.  In both~\S\ref{sec:1}
and~\S\ref{sec:6} we state carefully the domain of the Wick-rotated field
theories; that is, we specify precisely the sheaf of background fields.
Furthermore, it is important to evaluate theories in families of manifolds and
bordisms, and we make clear the sort of families we use (e.g., holomorphic or
smooth).  Finally, we do not delve into unitary structures in this paper,
though we point out where they enter, especially in~\S\ref{sec:6}.  (Unitarity
in the context of WZW conformal blocks remains an interesting issue;
see~\cite{Lo,BF} for recent works.)
 
The middle sections are mathematical expositions of tangential structures and
invertible field theory.  In~\S\ref{sec:2} after introducing tangential
structures in general, following Lashof~\cite{L}, we define $p_1$-structures.
\autoref{thm:11} relates local changes of various tangential structures.
\autoref{sec:4} takes up invertible field theories, mostly focusing on the
topological case in which an invertible field theory is a map out of a
Madsen--Tillmann spectrum.  In~\S\ref{sec:5} we turn to nontopological
invertible field theories, which are based on differential cohomology.  We
sketch the construction of the theory~$\gamma _c$, which physicists call
`gravitational Chern--Simons theory', a variation of a prime example
in~\cite{CS}. 
 
We thank Andr\'e Henriques, Mike Hopkins, and Greg Moore for discussions and
correspondence.

{\small
\def\reftext{References}
\renewcommand{\tocsection}[3]{%
  \begingroup 
   \def\tmp{#3}%
   \ifx\tmp\reftext
  \indentlabel{\phantom{1}\;\;} #3%
  \else\indentlabel{\ignorespaces#1 #2.\;\;}#3%
  \fi\endgroup}
\tableofcontents
}

   \section{Yang--Mills $+$ Chern--Simons}\label{sec:1}

We begin with the physics origin of Chern--Simons theory as a limit of a
Yang--Mills theory.  Namely, we define a family of nontopological
theories~$F_{e,\lambda }$ and take a singular limit $e\to \infty $.  The result
is (conjecturally) projectively topological, and we tensor with an invertible
field theory to obtain a linear topological theory~$Z_\lambda $.  We work in
Wick-rotated quantum field theory as axiomatized by Segal~\cite{S,KS}; see
also~\cite{We}.  A theory is a linear representation of a bordism category, and
to specify its flavor we use the following, specialized to dimension three.

  \begin{definition}[]\label{thm:12}
 \ 
 \begin{enumerate}[label=\textnormal{(\arabic*)}]

 \item $\Mant$ is the category whose objects are smooth 3-manifolds and whose
morphisms are local diffeomorphisms.

 \item A \emph{presheaf on $\Mant$} is a functor $\sF\:\Mant\op\to \Set$ to the
category of sets.

 \item A \emph{sheaf} is a presheaf that satisfies the usual gluing condition
for open covers~\cite{FH1}.

 \end{enumerate}
  \end{definition}

  \begin{remark}[]\label{thm:13}
 \ 
 \begin{enumerate}[label=\textnormal{(\arabic*)}]

 \item We encounter sheaves of groupoids as well as sheaves of sets.

 \item A \emph{locally constant} sheaf on~$\Mant$ is equivalent to a
3-dimensional tangential structure, as defined below in \autoref{thm:1}; see
\cite[\S24.4]{F1} for a sketch proof and additional exposition.

 \end{enumerate} 
  \end{remark} 

\noindent
 Examples of sheaves on $\Mant$ that we encounter map a ``test'' 3-manifold~$M$
to: 
  \begin{equation}\label{eq:25}
     \begin{aligned} \sF_{\Riem}\: M &\longmapsto \Riem(M),\qquad
      &&\textnormal{the set of Riemannian metrics}; \\
      \sF_{w_1}\:M&\longmapsto \Orient(M), \qquad &&\textnormal{the set of
      orientations}; \\ \sF_{w_1,w_2}\: M&\longmapsto \Spin(M), \qquad
      &&\textnormal{the groupoid of Spin structures}.\end{aligned} 
  \end{equation}
The latter two are locally constant sheaves.  The notation indicates the
Stiefel--Whitney class(es) that are trivialized.  We use Cartesian
products of these sheaves, such as $\sFRw=\sF_{\Riem}\times \sF_{w_1}$.

To each sheaf~$\sF$ corresponds a bordism category $\Btt(\sF)$.  An object is a
closed 2-manifold~$Y$ equipped with a cooriented embedding $Y\subset X$ in a
germ of a 3-manifold, together with an element of~$\sF(X)$.  That element, a
section of~$\sF$ over~$X$, is called a \emph{background field}.  A
(\!\emph{Wick-rotated\textnormal{)} 3-dimensional field theory over~$\sF$} is a
symmetric monoidal functor\footnote{The symmetric monoidal structures are
disjoint union~$\sqcup $ and tensor product~$\otimes $.}
  \begin{equation}\label{eq:26}
     F\:\bigl(\Btt(\sF),\sqcup \bigr)\longrightarrow \bigl(\tV,\otimes
     \bigr), 
  \end{equation}
where $\tV$ is the category of topological vector spaces and nuclear maps; see
\cite{S,KS,We,F1} for details and expanded exposition.  The theory~$F$ is
\emph{topological} if it factors through a locally constant sheaf~$\sF'$ via a
map $\sF\to \sF'$.

  \begin{remark}[]\label{thm:28}
 Definition~\eqref{eq:26} of a field theory only evaluates the theory on a
single manifold at a time.  This is not sufficient.  Rather, one should
``sheafify'' the definition over the category $\Man$ of smooth manifolds and
smooth maps.  In other words, one should evaluate a field theory on
\emph{families} of manifolds/bordisms parametrized by a smooth manifold~$S$.
(Stolz--Teichner emphasize this point in their survey~\cite[\S2.5]{ST}.)
For example, to a fiber bundle $Y\to S$ with fiber closed 2-dimensional
manifolds a 3-dimensional theory should assign a smooth complex vector bundle
over~$S$.  In a unitary theory this bundle will be equipped with a hermitian
metric and compatible covariant derivative.  There is also a holomorphic
version of this concept.  We encounter all of these variations
in~\S\ref{sec:6}. 
  \end{remark}

  \subsection{The physics origin of Chern--Simons theory}\label{subsec:1.1}

One believes~\cite{W3} that for each compact Lie group~$G$ there exists a
2-parameter family of Wick-rotated 3-dimensional field theories
  \begin{equation}\label{eq:27}
     F_{e,\lambda }\:\Btt(\sFRw)\longrightarrow \tV 
  \end{equation}
called \emph{Yang--Mills $+$ Chern--Simons} (abbreviated \YC).  (The
Chern--Simons term was introduced into 3-dimensional gauge theory
in~\cite{Sch,DJT}.)  Here $e\in \Rp$ is the \emph{coupling constant} and
$\lambda $~is a (cocycle for a) class in $H^4(BG;\ZZ)$ called the \emph{level}.
Fix a positive definite $G$-invariant inner product~$\langle -,- \rangle$
on~$\mathfrak{g}$.  There is a semiclassical description as a gauge theory with
Lagrangian
  \begin{equation}\label{eq:28}
     L_{e,\lambda }(A) = \frac{1}{4e^2}\langle F_A\wedge
     *F_A\rangle\;+\;\Gamma _\lambda (A),
  \end{equation}
where $A$~is a \emph{fluctuating} $G$-connection and $\Gamma _\lambda $~is the
Chern--Simons ``form''.  The Yang--Mills term uses the Riemannian metric; the
Chern--Simons term uses the orientation.  The Chern--Simons ``form'' is more
properly a differential cocycle, as we indicate in~\S\ref{sec:5}.
 
Assuming the level~$\lambda $ is nondegenerate,\footnote{The real image
$\lambda _{\RR}\in H^4(BG;\RR)$ is equivalent to a symmetric bilinear form on
the Lie algebra~$\mathfrak{g}$, and it is the nondegeneracy of this form to
which we refer.} the theory~$F_{e,\lambda }$ is thought to be \emph{gapped}.
One piece of evidence: if $G$~is abelian, then one can do an easy analysis of
the free relativistic theory in Minkowski spacetime and compute the gap
explicitly; see~\cite[Homework~FP4]{W2}.  The expected gap leads to an expected
singular long distance limit which is a linear theory whose projectivization is
a \emph{topological} theory.  (See~\cite{F2,vD} for more on projective field
theories and anomalies; the lecture notes \cite[\S36]{F1} contain a more
expansive exposition of this story.)  These expectations are laid out in the
following, inspired by~\cite{W3} and~\cite[\S2]{W1}; see
also~\cite[\S\S2.2.\{3,10,17\}]{Mo} for $G=\U_1$.

  \begin{conjecture}[]\label{thm:14}
 \ 
 \begin{enumerate}[label=\textnormal{(\arabic*)}]

 \item The Lagrangian~\eqref{eq:28} determines a family of Wick-rotated field
theories. 

 \item The singular limit $e\to \infty $ exists and defines a field theory 
  \begin{equation}\label{eq:29}
     \Fl\:\Btt(\sFRw)\longrightarrow \tV. 
  \end{equation}

 \item The projectivization~$\bFl$ factors through~$\Btt(\sFw)$, hence it is a
projective \emph{topological} field theory.

 \item The projectivity (anomaly)~$\al$ of~$\bFl$ extends to~$\tal$, a
4-dimensional invertible field theory of oriented manifolds whose partition
function on a closed oriented 4-manifold~$W$ is
  \begin{equation}\label{eq:61}
     \tal(W) = \exp\left( \frac{2\pi ic(\lambda )}{24}\,\langle p_1(W),[W]
     \rangle\right) , 
  \end{equation}
where $c=c(\lambda )\in \QQ$ is the \emph{central charge} associated to the
level~$\lambda $.  \textnormal{(}The formula for~$c(\lambda )$ is~\eqref{eq:59}
below.\textnormal{)}

 \end{enumerate} 
  \end{conjecture}

  \begin{remark}[]\label{thm:24}
 \ 
 \begin{enumerate}[label=\textnormal{(\arabic*)}]

 \item Statement~(3) implies that the states spaces of~$\Fl$ are all finite
dimensional.

 \item The projectivity, or anomaly, of a 3-dimensional quantum field theory is
a once-categorified invertible 3-dimensional theory~\cite{F2}.  In this case it
is topological---it factors through a theory over~$\sFw$---and as stated in
\autoref{thm:14}(4) the anomaly theory~$\al$ extends to a 4-dimensional
invertible theory~$\tal$.  See~\S\ref{sec:4} for more on invertible theories.

 \item We discuss the central charge in~\S\ref{subsec:6.1}.  In particular, the
formula for ~$c(\lambda )$ in terms of~$\lambda $ is~\eqref{eq:59}.  From this
formula it is manifest that $c(\lambda )$~is a rational number.  
 
 \item The anomaly theory~$\al$ and its extension~$\tal$ only depend on
$c(\lambda )\pmod{24}$.
 
 \item The first Pontrjagin class enters the story in~\eqref{eq:61}, and this
makes the choice of $\wpo$-structures later on more natural than other
tangential structures.

 \item The topological projective theory~$\bFl$ is essentially the picture of
Chern--Simons theory put forward by Walker~\cite{Wa} long ago; see~\cite{H} for
a modern account.  Walker uses the quantum group construction of the
topological field theory; see \autoref{thm:15}(5) below.

 \item A version of this conjecture for $G$~a torus is discussed in detail
in~\cite[\S4]{GMMS}.  

 \end{enumerate}
  \end{remark}

Next we explain Witten's maneuver~\cite[(2.23)]{W1} to construct a lift of the
projective theory~$\bFl$ to a linear \emph{topological} field theory.  (The
singular limit~$F_\lambda $ is a linear \emph{nontopological} lift of~$\bFl$.)
This uses another locally constant sheaf, $\sF_{p_1}$, equivalent to the
tangential structure built from trivializing the first Pontrjagin class; we
introduce it in~\S\ref{subsec:2.3}.\footnote{In this context $p_1$-structures
first appeared in~\cite{BHMV}.  Atiyah~\cite{A} used ``2-framings'' instead of
$p_1$-structures.  We will also consider framings and stable framings
in~\S\ref{sec:2}.  But, as already stated, $p_1$-structures are most natural
here.}  The maneuver uses the commutative diagram
  \begin{equation}\label{eq:30}
     \begin{gathered} \xymatrix@C-1pc{&\sFRwp\ar[dl]\ar[dr] \\ \sFRw \ar[dr]
     && \sFwp\ar[dl] \\ &\sFw} \end{gathered}
  \end{equation}
of sheaves on $\Mant$.  (The maps are projections off of Cartesian products.)
The key ingredient is a family of \emph{invertible} field theories 
  \begin{equation}\label{eq:31}
     \gamma _c\:\Btt(\sFRwp)\longrightarrow \Line_{\CC},\qquad c\in \RR, 
  \end{equation}
based on the intrinsic\footnote{Physicists call this the ``gravitational
Chern--Simons invariant''.  It is a secondary invariant associated to the first
Pontrjagin class of the \emph{tangent bundle}.  We sketch the construction
in~\S\ref{sec:5}.} Chern--Simons invariant of a Riemannian
3-manifold~\cite{CS}.  Witten's observation is that the theory
  \begin{equation}\label{eq:32}
     \Zl = \Fl\otimes \gamma _{-c(\lambda )} ,
  \end{equation}
initially defined over~$\sFRwp$, descends to a linear topological field theory
over~$\sFwp$: the metric dependence cancels out.   

  \begin{remark}[]\label{thm:15}
 \ 
 \begin{enumerate}[label=\textnormal{(\arabic*)}]

 \item The projective theory $\bFl\bigm/ \sFw$ has two linear lifts in this
story: (i)~the long distance limit $\Fl\bigm/\sFRw$ of \YC, and (ii)~the
topological field theory $\Zl\bigm/\sFwp$.

 \item If $\beta $~is an \emph{invertible} 3-dimensional topological field
theory of $\wpo$-manifolds, then $\Zl\otimes \beta $ is also a linear lift
of~$\bFl$.  As we will see in~\S\ref{subsubsec:4.1.3}, isomorphism classes of
such~$\beta $ form an abelian group isomorphic to~$\bmu6$.  This is one
explanation for the $\bmu6$-extension of the Witt group that appears
in~\cite{FST}.

 \item The invertible field theories in~(2) are not assumed to be unitary.  If
we impose unitarity, then the $\bmu6$~indeterminacy of the linear lift
of~$\bFl$ reduces to a $\bmu3$~indeterminacy; see~\eqref{eq:45}.  This is
closely related to the distinction between 3-dimensional $\wpo_3$-structures
and stable $\wpo_s$-structures that we make
in~\S\ref{subsec:2.3}; see~\cite{FH2}.

 \item The topological theory~$\Zl$ admits a unitary structure, as do~$\Fl$
and~$\gamma _{-c(\lambda )}$.  However, the isomorphism~\eqref{eq:32} need not
preserve these unitary structures.

 \item An extension of~$\Zl$ to a (1,2,3)-theory is what is usually called
(quantum) \emph{Chern--Simons theory} in topology.  It was constructed
rigorously by Reshetikhin--Turaev~\cite{RT2,Tu} starting from quantum group
data.  More generally, one can begin with a modular tensor category and
construct a (1,2,3)-\emph{Reshetikhin--Turaev Theory}.  In ~\cite{FST} we
construct a fully local\footnote{a (0,1,2,3)-theory} extension of
Reshetikhin--Turaev theories using the cobordism hypothesis.  This applies
to~$\Zl$. Of course, we can then tensor with~$\gamma _{c(\lambda )}$ to
rigorously construct what should be the long distance limit of \YC\ as a fully
extended theory.

 \item The topological field theory $\Zl\bigm/\sFwp$ only sees $c\pmod{24}$.
There are further statements along these lines in~\cite{FST}.  On the other
hand, $\gamma _c$~requires a specification of~$c\in \RR$, hence
by~\eqref{eq:32} so too does the limit $\Fl\bigm/\sFRw$ of \YC.
 
 \item For all $N\in \ZZ$, the theory~$\gamma \mstrut _{24N}$ does not depend
on a $p_1$-structure; it factors through~$\sFRw$.  The tensor product
$F_{e,\lambda }\otimes \gamma \mstrut _{24N}$ has the same underlying
projective theory as~$F_{e,\lambda }$, hence so too the same singular limit
$e\to \infty $ as a projective theory.  This is another reason why we only
expect the topological field theories constructed from \YC\ to have a central
charge in~$\QQ/24\ZZ$ rather than in~$\QQ$.
 
 \item We execute the Witten maneuver for a free spinor field
in~\S\ref{subsec:6.4}. 

 \end{enumerate}
 \end{remark}

  \subsection{The Spin variation}\label{subsec:1.2}

There are Chern--Simons theories whose level~$\lambda $ lies not in ordinary
cohomology, but rather in a slightly exotic cohomology theory.  Even on the
classical level these theories require a Spin structure, not just an
orientation.  Aspects of Spin Chern--Simons theory are studied
in~\cite{BM,B,J1,J2,BMo}.  In these theories the analog of the modular tensor
category is enriched over the category of super vector spaces and there are
other new features.  In~\cite{FST} we consider Spin theories as well as non-Spin
theories.  Here, in~\S\ref{sec:6}, we discuss another Spin theory: the free
2-dimensional spinor field and its anomaly, in various guises.
 
The key change over~\S\ref{subsec:1.1} is that~\eqref{eq:30} is replaced with
  \begin{equation}\label{eq:33}
     \begin{gathered} \xymatrix@C-1pc{&\sFRwwp\ar[dl]\ar[dr] \\ \sFRww \ar[dr]
     && \sFwwp\ar[dl] \\ &\sFww} \end{gathered} 
  \end{equation}
Each sheaf includes a Spin structure: a trivialization of both~$w_1$ and~$w_2$.
Now \YC\ is defined over Riemannian Spin manifolds, and the projective long
distance limit~$\bFl$ is defined over Spin manifolds.  The tangential structure
derived from~$\sFwwp$ is defined in~\S\ref{subsec:2.4}.

  \begin{remark}[]\label{thm:16}
 \ \begin{enumerate}[label=\textnormal{(\arabic*)}]

 \item The anomaly theory, as a theory of $\Spin_3$-manifolds, only depends on
$c(\lambda )\pmod6$.  (Compare \autoref{thm:24}(4).)

 \item The indeterminacy of \autoref{thm:15}(2) due to the cyclic group~$\bmu6$
of invertible $\wpo_3$-theories is now a $\bmu{48}\times \bmut$~indeterminacy
from invertible $\wwpo_3$-theories; see~\S\ref{subsubsec:4.1.4}.

 \item The theory~$\gamma _c$ is the same as in~\S\ref{subsec:1.1}: there is no
Spin structure dependence introduced. 

 \end{enumerate}
  \end{remark}

  \subsection{A formula for the central charge}\label{subsec:6.1}

The choice of~$c(\lambda )\in \QQ$---the central charge---in~\eqref{eq:32} is
given as follows, assuming that $\lambda $ is {positive definite} (as
defined shortly).  

Let $G$~be a compact Lie group and let $\lambda $ be a cocycle for a class in
$H^4(BG;\ZZ)$.  The real image $[\lambda _{\RR}]\in H^4(BG;\RR)$ is
equivalently a $G$-invariant symmetric bilinear form
  \begin{equation}\label{eq:56}
     \langle -,- \rangle_\lambda \:\mathfrak{g}\times
     \mathfrak{g}\longrightarrow \RR. 
  \end{equation}
We say that $\lambda $~is \emph{positive definite} if this form is positive
definite, which we now assume.  The negative of one-half the Killing form is a
nondegenerate $G$-invariant symmetric bilinear form
  \begin{equation}\label{eq:57}
     \langle -,- \rangle_h\:\mathfrak{g}\times \mathfrak{g}\longrightarrow
     \RR. 
  \end{equation}
Here $[2h]\in H^4(BG;\ZZ)$ is the first Pontrjagin class of the adjoint
representation.  Define the symmetric endomorphism $S_\lambda \in
\End(\mathfrak{g})$ by
  \begin{equation}\label{eq:58}
     \langle \xi _1,\xi _2 \rangle_\lambda = \langle \xi _1,S_\lambda (\xi
     _2) \rangle_{\lambda +h},\qquad \xi _1,\xi _2\in \mathfrak{g}. 
  \end{equation}
Then the central charge is the trace
  \begin{equation}\label{eq:59}
     c(\lambda ) = \tr(S_\lambda ). 
  \end{equation}
For example, if $G=\SU_2$ and the level~$\lambda $ is $k$~times the positive
generator of $H^4(BG;\ZZ)$, then \eqref{eq:59}~reduces to the usual formula
  \begin{equation}\label{eq:60}
     c(\lambda )= \frac{3k}{k+2}. 
  \end{equation}
Equation~\eqref{eq:59} can be derived from the Segal--Sugawara construction of
the boundary WZW~theory.  

   \section{Tangential structures}\label{sec:2}
 
Let $\BO_n$ be a choice of classifying space of the orthogonal group.  The
inclusions $\O_1\hookrightarrow \O_2\hookrightarrow \cdots$ induce maps
$\BO_1\longrightarrow \BO_2\longrightarrow \cdots$; the colimit of this
sequence of topological spaces is denoted~$\BO$.  The following definition was
introduced by Lashof~\cite{L}.

  \begin{definition}[]\label{thm:1}
 \ 
 \begin{enumerate}[label=\textnormal{(\arabic*)}]

 \item For $n\in \Zp$ an \emph{$n$-dimensional tangential structure} is a
continuous map $\pi_n\:\sX_n\to \BO_n$.

 \item A \emph{stable tangential structure} is a continuous map $\pi \:\sX\to
\BO$. 
 \end{enumerate}

  \end{definition}

\noindent
 Equivalently, an $n$-dimensional tangential structure is a
pair~$(\sX_n,\sE_n)$, where $\sX_n$~is a topological space and $\sE_n\to \sX_n$
is a rank~$n$ real vector bundle; there is a similar formulation in the stable
case.  A stable tangential structure induces an $n$-dimensional tangential
structure by pullback: 
  \begin{equation}\label{eq:1}
     \begin{gathered} \xymatrix{\sX_n\ar@{-->}[r]^{} \ar@{-->}[d]_{\pi_n} &
     \sX\ar[d]^{\pi} \\ \BO_n\ar[r]^{} & \BO} \end{gathered} 
  \end{equation}
Similarly, an $n$-dimensional tangential structure induces an $m$-dimensional
tangential structure for all~$m<n$. 

  \begin{definition}[]\label{thm:2}
  Suppose an $n$-dimensional tangential structure $\pi_n\:\sX_n\to \BO_n$ is
given.  
 \begin{enumerate}[label=\textnormal{(\arabic*)}]

 \item A \emph{$\pi_n$-structure} on a smooth $n$-manifold~$M$ is a lift~$s_n$
of a classifying map\footnote{The classifying map is a contractible choice: the
category of smooth manifolds and smooth maps is equivalent to a category whose
objects are pairs of a smooth manifold and a choice of classifying map for its
tangent bundle.} of its tangent bundle:
  \begin{equation}\label{eq:2}
     \begin{gathered} \xymatrix@C+1pc{&\sX_n\ar[d]^{\pi_n} \\
     M\ar[ur]^{s_n}\ar[r]^{\tau _M} & \BO_n} \end{gathered}
  \end{equation}
Two $\pi _n$-structures are \emph{isomorphic} if the lifts are homotopic.

 \item Suppose $\pi_n$~is a principal fibration with fiber~$F$.  A \emph{change
of $\pi_n$-structure} on~$M$ is a map $M\to F$.  A change of \emph{isomorphism
class} of $\pi_n$-structure is a homotopy class of maps $M\to F$.

 \end{enumerate}

  \end{definition}

\noindent
 There is a space of $\pi _n$-structures.  If $p\:\sX\to \BO$ is a stable
tangent structure, there are analogous definitions of a $\pi$-structure and a
change of $\pi$-structure on a smooth manifold.   

  \begin{lemma}[]\label{thm:25}
 Let $\pi_n\:\sX_n\to \BO_n$ be induced from $\pi\:\sX\to \BO$, as
in~\eqref{eq:1}, and suppose $M$~is an $n$-manifold.  Then there is a
homeomorphism between the space of $\pi_n$-structures on~$M$ and the space of
$\pi$-structures on~$M$. 
  \end{lemma}

  \begin{proof}
 In the diagram
  \begin{equation}\label{eq:3}
     \begin{gathered} \xymatrix@C+=3pc{& \sX_n\ar[r]^{} \ar[d]_{\pi_n} &
     \sX\ar[d]^{\pi} \\ M\ar@{-->}[ur]^{s_n}
     \ar@{-->}[urr]^<<<<<<<<<<<<<<<<<<<<<<<<s \ar[r]^<<<<<<<<<{\tau
     _M}&\BO_n\ar[r]^{} & \BO} \end{gathered}
  \end{equation}
the right hand square is a pullback.
  \end{proof}

  \begin{remark}[]\label{thm:4}
 \ 
 \begin{enumerate}[label=\textnormal{(\arabic*)}]

 \item For some tangential structures more rigid models are possible, such as
for framings, orientations, and Spin structures.  Rigid models---reductions of
the principal bundle of frames of a smooth $n$-manifold---exist for
$n$-dimensional tangential structures $BG_n\to \BGL_n\!\RR$ that are induced
from a homomorphism $G_n\to \GL_n\!\RR$ of \emph{Lie groups}.  (In the
preceding we have replaced~$\GL_n\!\RR$ by the homotopy equivalent maximally
compact subgroup~$\O_n\subset \GL_n\!\RR$.)

 \item In \autoref{thm:2}(2) the entire $\pi_n$-structure is changed; all that
is left fixed is the underlying smooth manifold.  We will later consider
changes of structure in which an underlying orientation or Spin structure is
left fixed.

 \end{enumerate}
 \end{remark}

For the purposes of this paper, we set~$n=3$ and introduce the relevant
tangential structures.

  \subsection{Framings}\label{subsec:2.1}

Let $*$~denote the singleton topological space.  A \emph{framing} or
\emph{3-framing} is the tangential structure\footnote{One can replace $*$ with
a contractible space~$E\!\O_3$ on which $\O_3$~acts freely.} 
  \begin{equation}\label{eq:4}
     \fr_3\:*\longrightarrow \BO_3 
  \end{equation}
A \emph{stable framing} is the tangential structure 
  \begin{equation}\label{eq:5}
     \fr\:*\longrightarrow \BO 
  \end{equation}
Observe that \eqref{eq:4}~is truly unstable: while it is induced from a
4-dimensional tangential structure, it is not induced from a 5-dimensional
tangential structure.\footnote{If it were, then in the 5-dimensional tangential
structure $\sX_5\to \BO_5$ the space $\sX_5$~would be a delooping of the
Stiefel manifold $O_5/O_3$, and this can be proved not to exist
using~\cite{Br}.} A 3-framing on a 3-manifold~$M$ is also called a
\emph{parallelism}.  A change of parallelism (\autoref{thm:2}(2)) is effected
by a map $M\to \O_3$; if it preserves the induced orientation then it is a map
$M\to \SO_3$.  This can be seen from the rigid model of~$s$ as a global basis
of the tangent bundle, or from the diagram
  \begin{equation}\label{eq:6}
     \begin{gathered} \xymatrix@C+2pc{&\O_3\ar[d] \\ & \ast\ar[d]^{\fr_3} \\
     M\ar@{-->}[uur] \ar[ur]^<<<<<<<<<<<<<s \ar[r]^<<<<<<<<<<<{\tau _M} &
     \BO_3} \end{gathered}
  \end{equation}
A similar comment applies to stable framings: orientation-preserving changes
are maps $M\to \SO$.

  \begin{remark}[]\label{thm:5}
 A framing induces a stable framing, as in~\eqref{eq:3}, but not every stable
framing comes from a framing.  For example, $S^2$~admits stable framings
(unique up to homotopy), but the hairy ball theorem obstructs the existence of
framings.  On the other hand, $S^3$~admits 3-framings.
  \end{remark}

  \begin{lemma}[]\label{thm:6}
 The map 
  \begin{equation}\label{eq:7}
     \pi _3\SO_3\longrightarrow \pi _3\SO 
  \end{equation}
from orientation-preserving changes of framings of~$S^3$ to
orientation-preserving changes of stable framings of~$S^3$ is multiplication
by~2. 
  \end{lemma}

\noindent
 Each of $\pi _3\SO_3$, $\pi _3\SO$ is infinite cyclic; the lemma asserts that
\eqref{eq:7}~maps a generator of~$\pi _3\SO_3$ to twice a generator of~$\pi
_3\SO$.  

  \begin{remark}[]\label{thm:7}
 Homotopy groups are defined using \emph{pointed} maps $S^3\to \SO_3$ and
$S^3\to \SO$, whereas changes of (stable) framings are \emph{unpointed}.  A
straightforward argument proves that $\pi _3\SO_N\to [S^3,\SO_N]$ is bijective
for all~$N\in \ZZ^{\ge3}$. 
  \end{remark}

  \begin{proof}
 Compute~\eqref{eq:7} by passing to the double cover 
  \begin{equation}\label{eq:11}
     \pi _3\Spin_3\longrightarrow \pi _3\Spin 
  \end{equation}
Under the special isomorphisms of low dimensional Spin groups, the inclusions 
  \begin{equation}\label{eq:12}
     \xymatrix@1{\Spin_3\;\ar@{^{(}->}[r] & \;\Spin_4\;\ar@{^{(}->}[r]&
     \;\Spin_6}
  \end{equation}
become the inclusions 
  \begin{equation}\label{eq:13}
     \xymatrix@1{\SU_2\;\ar@{^{(}->}[r]^<<<<<\Delta & \;\SU_2\times
     \SU_2\;\ar@{^{(}->}[r]& \;\SU_4}
  \end{equation}
where $\Delta $~is the diagonal inclusion and the second map is the block
diagonal inclusion.  This composition represents twice a generator of~$\pi
_3\SU_4$, and we are well into the stable range.
  \end{proof}

  \subsection{Orientations and Spin structures}\label{subsec:2.2}

Each is a \emph{stable} tangential structure, as witnessed by the pullback
diagrams
  \begin{equation}\label{eq:8}
     \begin{gathered} \xymatrix{\BSO_3\ar[r]^{} \ar[d]_{\pi_3} & \BSO\ar[d]^{\pi}
     \\ \BO_3\ar[r]^{} & \BO}\qquad \qquad \xymatrix{\BSpin_3\ar[r]^{}
     \ar[d]_{\pi_3} & \BSpin\ar[d]^{\pi} \\ \BO_3\ar[r]^{} & \BO}
     \end{gathered} 
  \end{equation}
The fibers of the vertical maps are\footnote{$\bmu n$~is the cyclic group of
$n^{\textnormal{th}}$~roots of unity in~$\CC$.  $K(\pi ,q)$~is an
Eilenberg--MacLane space with $\pi _qK(\pi ,q)=\pi $.}  $\bmut=\{\pm1\}$ and
$\bmut\times B\!\bmut \;\sim\; \bmut\times K(\bmut,1)$, respectively.  Of
course, there are rigid geometric models.  (The paper~\cite{KT} is a useful
reference for Spin structures, especially in low dimensions.)
 
An orientation of a smooth manifold~$M$ is a trivialization of the orientation
double cover $\hM\to M$, and $\hM\to M$ represents the Stiefel--Whitney
class~$w_1$ of the tangent bundle $TM\to M$.  Hence we sometimes use
`$w_1$-structure' in place of `orientation'.  In a similar vein, we 
use `$(w_1,w_2)$-structure' synonymously with `Spin structure'. 

  \begin{remark}[]\label{thm:8}
 Stabilization induces an isomorphism of the set of orientations and the set of
stable orientations, as in \autoref{thm:25}.  Let $M$~be a 3-manifold.  The
corresponding equivalence of groupoids of (stable) Spin structures follows from
the diagram
  \begin{equation}\label{eq:9}
     \begin{gathered} \xymatrix{P\ar@{-->}[rr]^{} \ar@{-->}[d]^{\bmut} &&
     P'\ar[d]^{\bmut} \\ \SO(M)\ar@{^{(}->}[rr]^{} \ar[dr]^{\SO_3}
     && SO'(M)\ar[dl]_<<<<<<<<{\SO_N} \\ &M} \end{gathered}  
  \end{equation}
in which a stable Spin structure $P'\xrightarrow{\;\;\Spin_N\;\;}M$, $N\gg 3$,
pulls back to a Spin structure $P\xrightarrow{\;\;\Spin_3\;\;}M$ over the
3-manifold~$M$.  (This is the rigid model analog of~\eqref{eq:3}.)  By
contrast, recall (\autoref{thm:5}) that framings and stable framings differ.
  \end{remark}

  \subsection{$\wpo$-structures}\label{subsec:2.3}

The remaining tangential structures that we introduce do not admit rigid models
in the sense of \autoref{thm:4}.  A \emph{$p_1$-structure} is a trivialization
of the first Pontrjagin class.  Restrict to oriented manifolds.  Fix models
$\SZ k$, $k\in \Znn$, for Eilenberg--MacLane spaces with nonzero homotopy group
$\pi _k\cong \ZZ$.  Also fix maps that represent the first Pontrjagin
class~$p_1$, and define $\sX_3,\sX$ as the homotopy fibers of~$p_1$ in the
diagram
  \begin{equation}\label{eq:10}
     \begin{gathered} \xymatrix@C-1pc{\sX_3\ar@{-->}[rr]^{} \ar@{-->}[d]^{} &&
     \sX\ar@{-->}[d]^{} \\ \BSO_3\ar[rr]^{} \ar[dr]^<<<<<<<{p_1} &&
     \BSO\ar[dl]_<<<<<<<<{p_1} \\ &\SZ4} \end{gathered}
  \end{equation} 
After composition with $\BSO_3\to \BO_3$ and $\BSO\to \BO$ we obtain the
pullback diagram 
  \begin{equation}\label{eq:14}
     \begin{gathered} \xymatrix{\sX_3\ar[r]^{} \ar[d]_{\wpo\mstrut _3} &
     \sX\ar[d]^{\wpo_s} \\ \BO_3\ar[r]^{} & \BO} \end{gathered} 
  \end{equation}
in which the vertical maps define the stable tangential structure~$\wpo_s$ and
its pullback to a 3-dimensional tangential structure~$\wpo\mstrut _3$.  We also
call these the `stable oriented $p_1$-structure' (or `$\SO^{p_1}$-structure')
and the `3-dimensional oriented $p_1$-structure' (or
`$\SO^{p_1}(3)$-structure').  By \autoref{thm:25} there is an equivalence
between stable and unstable structures.
 
Let $M$~be an oriented 3-manifold.  A change of oriented $p_1$-structure on~$M$
that fixes the underlying orientation is a map
  \begin{equation}\label{eq:15}
     \begin{gathered} \xymatrix@C+=2.5pc{
     &\SZ3\ar@{=}[rr]\ar[d] && \SZ3\ar[d]\\
     &\sX_3\ar[rr]^{} \ar[d]^{} &&
     \sX\ar@{-->}[d]^{} \\ 
     M\ar@{-->}[uur] \ar[ur]^<<<<<<<<<s \ar[r]^<<<<<<<{\tau _M} &     
     \BSO_3\ar[rr]^{} \ar[dr]^{p_1} &&
     \BSO\ar[dl]_<<<<<<<<{p_1} \\ &&\SZ4} \end{gathered} 
  \end{equation}
(The dotted map in the diagram shows a change of 3-dimensional structure; a map
to the other~$\SZ 3$ is a change of stable structure.)  Here $\SZ3$~is the
homotopy fiber of $\sX_3\to \BSO_3$ and of $\sX\to \BSO$; since the lower
triangle homotopy commutes, the induced map $\SZ3\to \SZ3$ is a homotopy
equivalence.  The homotopy class of a change of (stable) oriented
$p_1$-structure on~$M$ is given by a class in~$H^3(M;\ZZ)$.  

  \begin{lemma}[]\label{thm:9}
 Let $M$~be an oriented 3-manifold.  The map 
  \begin{equation}\label{eq:16}
     H^3(M;\ZZ)\longrightarrow H^3(M;\ZZ) 
  \end{equation}
from changes of 3-dimensional oriented $p_1$-structures to change of stable
oriented $p_1$-structures is an isomorphism. 
  \end{lemma}

  \begin{proof}
 Diagram~\eqref{eq:15}. 
  \end{proof}

  \begin{remark}[]\label{thm:17}
 Fix an oriented 3-manifold~$M$.  Isomorphism classes of compatible
3-dimensional $p_1$-structures on~$M$ form a torsor~$\sP(M)$ over~$H^3(M;\ZZ)$.
If $M$~is closed and connected, then $\sP(M)$~is a $\ZZ$-torsor.  It is useful
to contemplate this and related torsors in some of what follows.
  \end{remark}

  \subsection{$\wwpo$-structures}\label{subsec:2.4}

These are also called `3-dimensional Spin $p_1$-structures' and 'stable Spin
$p_1$-structures'.  Replace $\SO_3$ with $\Spin_3$ and $\SO$ with~$\Spin$
in~\S\ref{subsec:2.3} to make the constructions.  For the analog
of~\eqref{eq:14} we use the diagram 
  \begin{equation}\label{eq:17}
     \begin{gathered} \xymatrix{\sY_3\ar[r]^{} \ar[d]_{\wwpo\mstrut _3} &
     \sY\ar[d]^{\wwpo_s} \\ \BO_3\ar[r]^{} & \BO} \end{gathered} 
  \end{equation}
Here $\sY_3$~is the homotopy fiber of $p_1\:\BSpin_3\to \SZ4$ and $\sY$~is the
homotopy fiber of $p_1\:\BSpin\to \SZ4$.  The Spin version of \autoref{thm:9}
holds.  

  \begin{remark}[]\label{thm:10}
 The first Pontrjagin class~$p_1$ is divisible by~2 on~$\BSpin$ and its
pullback to~$\BSpin_3$ is divisible by~4.  This leads to cousins of the
$\wwpo$-structures.  With evident notation, a $(w_1,\,w_2, \,p_1/4)\mstrut
_3$-structure is equivalent to a 3-framing, and a $(w_1,\,w_2, \,p_1/2)\mstrut
_3$-structure is equivalent to a stable framing on manifolds of
dimension~$\le3$.
  \end{remark}

  \begin{remark}[]\label{thm:18}
 Fix a Spin 3-manifold~$M$.  Then isomorphism classes of compatible
3-dimensional $p_1/4$-structures, $p_1/2$-structures, and $p_1$-structures form
$H^3(M;\ZZ)$-torsors; there are torsor morphisms 
  \begin{equation}\label{eq:34}
     \sP_{1/4}(M)\longrightarrow \sP_{1/2}(M)\longrightarrow \sP(M) 
  \end{equation}
compatible with the group homomorphisms 
  \begin{equation}\label{eq:35}
     H^3(M;\ZZ) \xrightarrow{\;\;2\;\;}H^3(M;\ZZ)
     \xrightarrow{\;\;2\;\;}H^3(M;\ZZ) . 
  \end{equation}
Therefore, the morphisms~\eqref{eq:34} are injective and each image has
index~2.

  \end{remark}

  \subsection{Maps of tangential structures}\label{subsec:2.5}

The relevant maps are encoded in the diagram 
  \begin{gather}\label{eq:18}
      \vcenter{\footnotesize
      \xymatrix@C+3.5pc{\ast \ar[r]\ar[ddr]|{\fr_3} \ar[drrr]
     & \sY_3 \ar[r]\ar[dd]|{\wwpo \mstrut _3}\ar[drrr] & \sX_3 \ar[ddl]|{\wpo
     \mstrut _3} \ar[drrr]\\ &&& \ast \ar[r]\ar[dr]|{\fr_s} & \sY
     \ar[r]\ar[d]|{\wwpo_s} & \sX \ar[dl]|{\wpo_s} \\ &\BO_3\ar[rrr] &&&\BO}}
  \end{gather}
As for \emph{local} changes of structures, which are changes on~$M=S^3$, we
have the following.  Let $\Delta (p)$~denote the infinite cyclic group of local
changes of structure\footnote{For oriented and Spin $p_1$-structures these are
changes of the $p_1$-structure that fix the underlying orientation or Spin
structure.}~$\pi$.

  \begin{proposition}[]\label{thm:11}
 Local changes of structure are related by the commutative diagram 
  \begin{equation}\label{eq:19}
     \begin{gathered} \xymatrix@R+.5pc@C+.5pc{\Delta (\fr)_3\ar[r]^<<<<<{4}
     \ar[d]^{2} & \Delta \wwpo\mstrut _3\ar[d]^{1} \ar[r]^<<<<<1& \Delta
     \wpo\mstrut _3 \ar[d]^1\\ \Delta (\fr)_s\ar[r]^<<<<<{2} & \Delta
     \wwpo_s\ar[r]^<<<<<1 & \Delta \wpo_s} \end{gathered}
  \end{equation}
of homomorphisms of infinite cyclic groups. 
  \end{proposition}

  \begin{proof}
 The vertical arrows are computed in \autoref{thm:6} and \autoref{thm:9},
together with the Spin version of the latter.  We compute the composite of the
top horizontal arrows from the diagram 
  \begin{equation}\label{eq:20}
     \begin{gathered} \xymatrix@C-1pc{\SO_3\ar[rr]^{f} \ar[d]^{} &&
     \SZ3\ar[d]^{} \\ \ast\ar[rr]^{f} \ar[dr]^{} &&
     \sX_3\ar[dl] \\ &\BSO_3} \end{gathered} 
  \end{equation}
This induces a pullback map of Leray spectral sequences for cohomology with
integer coefficients:  
  \begin{sseqdata}[name=first, cohomological Serre grading, classes = {draw = none},xscale=1.15]
     \class["\ZZ "](0,0)
     \class["a "](0,2)
     \class["b "](2,2)
     \class["\kappa  "](0,3)
     \class["W_3 "](3,0)
     \class["p_1 "](4,0)
     \d["d_3" {xshift=1.4pc, yshift=1.1pc}]3(0,2)
     \d["d_2" {xshift=1.4pc, yshift=1.1pc}]2(0,3)
  \end{sseqdata}  
 
  \begin{sseqdata}[name=second, cohomological Serre grading, classes = {draw = none},xscale=1.15]
     \class["\ZZ "](0,0)
     \class["\iota  "](0,3)
     \class["W_3 "](3,0)
     \class["p_1 "](4,0)
     \d["d_4" {yshift=2pc}]4(0,3)
  \end{sseqdata}  

\medskip
  \printpage[ name = first] \qquad
  \raisebox{6pc}{$\xleftarrow{\;\;\;f^*\;\;\;} $} \qquad
  \printpage[ name = second] 
\smallskip

\noindent
 Here $2a=2b=0$.  The class~$2\kappa $ survives to the $E_4$-page, and
$d_4(2\kappa )=p_1$.  Hence
  \begin{equation}\label{eq:21}
     d_4f^*(\iota ) = f^*(d_4\iota ) = f^*(p_1) = p_1, 
  \end{equation}
from which $f^*(\iota) =2\kappa $.  This implies that $f_*\:H_3(\SO_3)\to
H_3(\SZ3)$ is multiplication by~2.  Since the Hurewicz homomorphism $\pi
_3\SO_3\to H_3(\SO_3)$ is also multiplication by~2, the composite of the top
horizontal arrows is multiplication by~4.  (Recall from \autoref{thm:6} that
local changes of oriented framings are computed in~$\pi _3\SO_3$.)
 
A similar argument shows that the composition of the bottom horizontal arrows
is multiplication by~2; the Hurewicz homomorphism $\pi _3\SO\to H_3(\SO)$ is an
isomorphism.  (Logically, we do not need this step in the proof.) 
 
We claim that the rightmost horizontal arrows in~\eqref{eq:19} are
isomorphisms.  For the arrow in the top row we use the diagram 
  \begin{equation}\label{eq:23}
     \begin{gathered} \xymatrix@C-1pc{\SZ3\ar[rr]^{h} \ar[d]^{} &&
     \SZ3\ar[d]^{} \\ \sY_3\ar[rr]^{h} \ar[d]^{} && \sX_3\ar[d] \\
     \BSpin_3\ar[rr]^h&&\BSO_3} 
     \end{gathered} 
  \end{equation}
The induced map of spectral sequences is now 
 
  \begin{sseqdata}[name=third, cohomological Serre grading, classes = {draw = none},xscale=1.15]
     \class["\ZZ "](0,0)
     \class["\iota ' "](0,3)
     \class["4\mu  "](4,0)
     \d["d_4" {yshift=2pc}]4(0,3)
  \end{sseqdata}  

\medskip
  \printpage[ name = third] \qquad
  \raisebox{6pc}{$\xleftarrow{\;\;\;h^*\;\;\;} $} \qquad
  \printpage[ name = second] 
\smallskip

\noindent 
 where $4\mu =h^*(p_1)$.  (In the first spectral sequence the class~$\mu $
generates the infinite cyclic group~ $H^4(\BSO_3;\ZZ)$.)  Then
  \begin{equation}\label{eq:24}
     d_4h^*(\iota ) = h^*(d_4\iota ) = h^*(p_1) = 4\mu = d_4(\iota '), 
  \end{equation}
from which $h^*(\iota )=\iota '$.  The argument for the bottom row (stable
analog) is the same except that $p_1$~is only divisible by~2
in~$H^4(\BSpin;\ZZ)$.
  \end{proof}

  \begin{remark}[]\label{thm:19}
 The two leftmost horizontal arrows in~\eqref{eq:19} can also be computed
using~\autoref{thm:10} and \autoref{thm:18} for~$M=S^3$. 
  \end{remark}


   \section{Invertible field theories}\label{sec:4}

The basic notion is straightforward, stated here in the 3-dimensional case.  We
use notation from~\eqref{eq:26}. 

  \begin{definition}[]\label{thm:20}
 A 3-dimensional field theory 
  \begin{equation}\label{eq:36}
     A\:\Btt(\sF)\longrightarrow \tV, 
  \end{equation}
is \emph{invertible} if it factors through the Picard groupoid of complex
lines: 
  \begin{equation}\label{eq:37}
     \begin{gathered} \xymatrix@R-1pc@C+1pc{&\Line\ar[dd] \\
     \Btt(\sF)\ar@{-->}[ur]^<<<<<<<<{\widehat\alpha } \ar[dr]^<<<<<<{A} \\ & \tV}
     \end{gathered} 
  \end{equation}
   \end{definition}

\noindent
 The factorization is unique if it exists.  As is true for general field
theories, an invertible field theory may or may not be topological.  In both
the topological and nontopological cases there is a ``homological''
interpretation.  And in both cases we consider \emph{fully local} field
theories, defined on a bordism 3-category $\Bord_{\langle 0,1,2,3
\rangle}(\sF)$.  In this section we consider topological theories.
In~\S\ref{sec:5} we describe the family~\eqref{eq:31} of nontopological
invertible theories.

First, we examine the arrow~$\widehat\alpha $ in~\eqref{eq:37} more closely.
It is a symmetric monoidal functor whose codomain is a Picard groupoid.
Therefore, $\widehat\alpha $~factors through the \emph{Picard groupoid
quotient} of the domain (also known as the \emph{group completion}): 
  \begin{equation}\label{eq:38}
     \begin{gathered}
     \xymatrix@R-1pc@C+1pc{\Btt(\sF)\ar[dd]\ar[rd]^{\widehat\alpha } \\
     &\Line \\|\Btt(\sF)|\ar@{-->}[ur]^<<<<<<<<{\alpha } } \end{gathered}
  \end{equation}
The geometric realization of a Picard groupoid is an infinite loop space, or
equivalently a connective spectrum.  For example, the geometric realization of
the Picard groupoid $\Line$ is the Eilenberg--MacLane spectrum~$\Sigma H\zt$.
The geometric realization of a topological bordism category was determined
in~\cite{GMTW}; it is a \emph{Madsen--Tillmann spectrum}~\cite{MT}.  So too is
the geometric realization of a fully local higher topological bordism
category~\cite{S-P}.  In the fully local case the codomain is also replaced by
a richer spectrum.  There is a universal choice~$\Sigma ^3\ICx$, which is dual
to the sphere spectrum.  Bottom line: a 3-dimensional invertible topological
field theory is a map from a Madsen--Tillmann spectrum to~$\Sigma ^3\ICx$.
See~\cite{FH2} and the references therein for a more generous account.

  \begin{remark}[]\label{thm:21}
 \ 
 \begin{enumerate}[label=\textnormal{(\arabic*)}]

 \item Other choices of codomain are possible, but $\Sigma ^3\ICx$ enjoys a
universal property: the abelian group of isomorphism classes of 3-dimensional
invertible topological field theories with codomain~$\Sigma ^3\ICx$ is
isomorphic to the Pontrjagin dual group to $\pi _3$ of the appropriate
Madsen--Tillmann spectrum.  More colloquially, with codomain~$\Sigma ^3\ICx$
the partition function determines the invertible theory.

 \item Unitarity is not included in the definition~\eqref{eq:26} of a
Wick-rotated quantum field theory.  In the invertible case, if we impose
unitarity then the Madsen--Tillmann spectrum is replaced by a stabilized
version~\cite{FH2}.  
 \end{enumerate}
  \end{remark}

  \subsection{Some invertible topological field theories}\label{subsec:4.1}

We rely on the bordism computations in \cite[Appendix~B]{FST}.  In the cases we
consider, the corresponding groups of invertible topological field theories are
sums of finite cyclic groups.  Here we tell the partition function of
generating theories.

   \subsubsection{Framed theories}\label{subsubsec:4.1.1}
 The Madsen--Tillmann spectrum built from 3-framed 3-manifolds is the sphere
spectrum~$\SS$.  The relevant homotopy group is the \emph{3-stem} 
  \begin{equation}\label{eq:39}
     \pi _3\SS \cong \zmod{24}. 
  \end{equation}
The Lie group~$\SU_2$ with left-invariant framing represents a generator. 
 
By the general remarks above, the group of 3-framed invertible 3-dimensional
topological field theories is isomorphic to the Pontrjagin dual~$\bmu{24}$ to
$\zmod{24}$.  A generating theory\footnote{In~\cite{FST} this is called the
``topological free fermion theory''.}~$\psi $ can be constructed using the
\emph{Adams $e$-invariant}~\cite{Ad,APS} as follows.  Let $X$~be a closed
3-framed 3-manifold.  Choose\footnote{Every closed Spin 3-manifold bounds a
compact Spin 4-manifold.  Similarly, every closed oriented 3-manifold bounds a
compact oriented 4-manifold.  A 3-framing on a 3-manifold fixes a Spin
structure.}~$W$ a compact Spin 4-manifold with $\partial W=X$.  The 3-framing
on~$X$ determines a relative first Pontrjagin class $p_1(W,X)\in H^4(W,X;\ZZ)$.
The partition function is
  \begin{equation}\label{eq:40}
     \psi (X) = \exp\left( \frac{2\pi i}{48}\bigl\langle p_1(W,X),[W,X]
     \bigr\rangle \right) , 
  \end{equation}
where $[W,X]\in H_4(W,X)$ is the relative fundamental class of the orientation.
(Observe that $p_1\in H^4(\BSpin;\ZZ)$ is divisible by two.)

   \subsubsection{Spin theories}\label{subsubsec:4.1.2}
 Whereas every closed Spin 3-manifold bounds a compact Spin 4-manifold, if we
impose a 3-dimensional reduction of structure on the tangent bundle of the
4-manifold this is no longer true:
  \begin{equation}\label{eq:41}
     \pi _3\Sigma ^3\MTSpin_3\cong \zt. 
  \end{equation}
(See~\cite[\S6]{F3} for exposition of this kind of bordism and Madsen--Tillmann
spectra.)  The 3-sphere with its unique Spin structure represents the
generator.  The nontrivial invertible 3-dimensional Spin theory~$\nu $ has
the following partition function.  Let $X$~be a closed Spin 3-manifold.  Write
$X=\partial W$ for a compact Spin 4-manifold~$W$.  Then 
  \begin{equation}\label{eq:42}
     \nu (X) = \exp\left( \frac{2\pi i}{2}\Euler(W)\right) , 
  \end{equation}
where $\Euler(W)$~is the Euler number of~$W$.  If $X$~is 3-framed, then
$\nu (X)=\psi (X)^{12}$. 

  \begin{remark}[]\label{thm:22}
 For $X=S^3$ with its unique Spin structure, we can choose $W=D^4$ to be the
4-disk, and so $\nu (S^3)=-1$.  This implies that $\nu $~does \emph{not} admit
a reflection positive structure, since $S^3$~is a double and the partition
function of a double is positive in a reflection positive theory
\cite[\S4.4]{FH2}. 
  \end{remark}

   \subsubsection{$\wpo$-theories}\label{subsubsec:4.1.3}
 For rank~3 structures the relevant bordism group is\footnote{The notation
reflects `$\wpo$-structure' = `oriented $p_1$-structure'.  The topological
group~$\SO^{p_1}_3$ is defined as the loop space of the homotopy fiber of the
map $p_1\:\BSO_3\to K(\ZZ,4)$.  We use notations `$\wpo_3$' and `$\wpo_s$' to
distinguish the rank~3 structure from the stable structure, as
in~\S\ref{subsec:2.3}.}
  \begin{equation}\label{eq:43}
     \pi _3\Sigma ^3\MTSO_3^{p_1}\cong \zmod6. 
  \end{equation}
The left parallelism on~$\SU_2$ induces a $\wpo_3$-structure, which represents
a generator of~\eqref{eq:43}.  In other words, the map $\pi _3\SS\to \pi
_3\Sigma ^3\MTSO_3^{p_1}$ is surjective.  A generator~$\beta $ of the
group~$\bmu6$ of invertible 3-dimensional $\wpo_3$-theories has the following
partition function.  Let $X$~be a closed $\wpo$ 3-manifold, and write
$X=\partial W$ for a compact oriented 4-manifold~$W$.  Then
  \begin{equation}\label{eq:44}
     \beta (X) = \exp\left( \frac{2\pi i}{6}\Bigl\{ \bigl\langle p_1(W,X),[W,X]
     \bigr\rangle + 3\Euler(W) \Bigr\} \right) . 
  \end{equation}
If $X$~is 3-framed, then $\beta (X)=\psi (X)^{-4}$. 
 
Turning now to \emph{stable} $\wpo_s$-structures, the relevant bordism group is 
  \begin{equation}\label{eq:45}
     \pi _3\MTSO^{p_1}\cong \zmod3; 
  \end{equation}
the order 3 theory~$\beta ^2$ factors to a $\wpo_s$-theory that generates the
group~$\bmu3$ of invertible 3-dimensional $\wpo_s$-theories.

   \subsubsection{$\wwpo$-theories and variations}\label{subsubsec:4.1.4}
 Here we meet noncyclic bordism groups.  Namely, the relevant bordism group of
rank~3 $\wwpo_3$-structures is
  \begin{equation}\label{eq:46}
     \pi _3\Sigma ^3\MTSpin_3^{p_1}\cong \zmod{48}\;\oplus\; \zt, 
  \end{equation}
whereas the relevant bordism group of stable $\wwpo_s$-structures is 
  \begin{equation}\label{eq:47}
     \pi _3\MTSpin^{p_1} \cong \zmod{48}. 
  \end{equation}
Fix an orientation of~$S^3$ and its (unique up to isomorphism) refinement to a
Spin structure.  Let $\sF_s(S^3)$ denote the $\ZZ$-torsor of oriented stable
framings of~$S^3$, and recall from \autoref{thm:18} the $\ZZ$-torsor $\sP(S^3)$
of $p_1$-structures compatible with the Spin structure.  Then by
\autoref{thm:11} the map
  \begin{equation}\label{eq:48}
     \sF_s(S^3)\longrightarrow \sP(S^3) 
  \end{equation}
has image of index~$2$.  Let $\fp\in \sP(S^3)$ be the image of the Lie group
framing on~$S^3\cong \SU_2$.  Then $\fp-1\in \sP(S^3)$ generates $\pi
_3\MTSpin^{p_1}$.  Furthermore, it represents an element~$x$ of order~48 in
$\pi _3\Sigma ^3\MTSpin_3^{p_1}$.  The Spin 3-sphere~$S^3$ with a bounding
$p_1$-structure represents an element~$y$ of order~2 in $\pi _3\Sigma
^3\MTSpin_3^{p_1}$ that does not equal~$24x$.  The elements~$x,y$ together
generate~\eqref{eq:46}.  

  \begin{remark}[]\label{thm:23}
 The preceding follows from \cite[(B.7)]{FST}, in which the natural map 
  \begin{equation}\label{eq:49}
     \pi _3\Sigma ^3\MTSpin_3^{p_1} \longrightarrow \pi _3\Sigma
     ^3\MTSpin^{p_1}\;\oplus \;\pi _3\Sigma ^3\MTSpin_3
  \end{equation}
is shown to be the isomorphism~\eqref{eq:46}. 
  \end{remark}

Turning to the corresponding invertible field theories, whose isomorphism
classes form a group isomorphic to $\bmu{48}\times \bmut$, the pullback of~$\nu
$ to~$\Sigma ^3\MTSpin_3^{p_1}$ generates the $\bmut$~factor.  A
generator~$\lambda $ of the $\bmu{48}$~factor has the following partition
function.  Let~$X$ be a closed Spin 3-manifold equipped with a $p_1$-structure.
Choose a compact Spin 4-manifold~$W$ with $\partial W=X$.  Then 
  \begin{equation}\label{eq:50}
     \lambda (X) = \exp\left( \frac{2\pi i}{48}\bigl\langle p_1(W,X),[W,X]
     \bigr\rangle \right) . 
  \end{equation}
Comparing with~\eqref{eq:40} we see that $\lambda $~is the extension of~$\psi $
to Spin $p_1$-manifolds.

   \section{Tangential Chern--Simons theory}\label{sec:5}

In this section we construct the invertible 3-dimensional field theory~$\gamma
_c$, $c\in \RR$, that was introduced in~\eqref{eq:31}.  This theory is
\emph{nontopological}, so it is \emph{not} described by a map out of a
Madsen--Tillmann spectrum.  Rather, it is constructed using the theory of
\emph{differential cohomology}.  We begin with a lightning review of ordinary
differential cohomology.  Then we construct the theory~$\gamma _c$
of~\eqref{eq:31}, which recall is an invertible 3-dimensional theory over the
sheaf $\sFRwp\to \Man_3$.

  \subsection{Differential cocycles}\label{subsec:5.1}

A comprehensive treatment together with extensive history and referencing may
be found in~\cite{ADH}; see also the survey~\cite{D} and the references
therein.  A foundational paper is~\cite{HS}.
 
As in~\S\ref{sec:1} we work with sheaves on~$\Man_3$, though in this subsection
all sheaves extend to~$\Man$, the category of smooth manifolds and smooth maps
between them.  (See~\cite{FH1} for background.)  There are sheaves~$\Omega ^q$,
$q\in \ZZ$, whose sections over a smooth manifold~$M$ form the vector
space~$\Omega ^q(M)$ of differential $q$-forms on~$M$.  There are also
sheaves~$\Omega ^q_{\textnormal{cl}}$ of closed differential forms of
degree~$q$.  As well, there are sheaves~$Z^q$ of smooth singular cocycles with
integer coefficients.  The sheaf~$\hZ$ of \emph{differential cocycles} is
constructed as a homotopy pullback in the references at the beginning of this
subsection.  A differential cocycle has a curvature, which is a closed
differential form, and also an underlying smooth singular cocycle; this is
expressed in the diagram of sheaves
  \begin{equation}\label{eq:51}
     \begin{gathered} \xymatrix{\hZ^q\ar[r]^{\omega } \ar[d]^{\mu } & \Omega
     ^q_{\textnormal{cl}}\\ Z^q} \end{gathered} 
  \end{equation}
Suppose $M\in \Man$, $\hc\in \hZ^q(M)$, and $\tau \in C^{q-1}(M)$ is a
trivialization of~$\mu (\hc)$ in the sense that $\delta \tau =\mu (\hc)$.  Then
$\hc$~lifts\footnote{This construction is spelled out in a particular cochain
model in~ \cite[\S A.2]{FN}.  In the notation of that reference, if
$\hc=(c,h,\omega )$ is a cocycle in~$\hC(q)^q(M)$, then $\mu (\hc)=(c,h,\omega
)$~is its image in~$\hC(q-1)^q(M)$.  The trivializing cochain of~$\mu (\hc)$
has the form $\tau =(b,k,\phi )\in \hC(q-1)^{q-1}(M)$, and $\phi $~is the
desired $(q-1)$-form.} uniquely to a differential $(q-1)$-form on~$M$.  In
terms of sheaves there is a fibration
  \begin{equation}\label{eq:52}
     \begin{gathered} \xymatrix{\Omega ^{q-1}\ar[d]\\ \hZ^q\ar[d]^{\mu } \\
     Z^q } \end{gathered} 
  \end{equation}

The preceding is based on integral Eilenberg--MacLane cohomology~$H\ZZ$.  There
are differential refinements of generalized cohomology theories as well.
Particular cases are used to define (classical) Spin Chern--Simons invariants;
see~\cite{J1,FN}, for example.

  \subsection{The theory~$\gamma _c$}\label{subsec:5.2}

We first recast the classical Chern--Simons form~\cite{CS} as a differential
cocycle.  Let $G$~be a Lie group with finitely many components, and fix a
cocycle~$\lambda $ for a class (the \emph{level}) in~$H^4(BG;\ZZ)$.  Let
$\sFBN{G}$ denote the simplicial sheaf of $G$-connections.  The
\emph{Chern--Simons differential cocycle} is a map~$\Gamma $ that fits into the
diagram
  \begin{equation}\label{eq:53}
     \begin{gathered} \xymatrix@C+1pc{\sFBN G \ar[r]^{\Gamma } & \hZ^4\ar[r]^\omega
     \ar[d]^{\mu } & \Omega ^4_{\textnormal{cl}} \\ & Z^4 } \end{gathered} 
  \end{equation}
where $\omega \circ \Gamma $~is the Chern--Weil 4-form of the level~$\lambda $,
and we can assume that $\mu \circ \Gamma =\lambda $. 
 
Fix $G=\SO_3$ and let the level be the first Pontrjagin class $p_1\in
H^4(\BSO_3;\ZZ)$.  In the following diagram $\TLC$~is the Levi--Civita
connection: 
  \begin{equation}\label{eq:54}
     \begin{gathered} \xymatrix@C+1pc{\sF_{\Riem,p_1} \ar@{-->}[rr]^\phi \ar[d]
     &&\Omega ^3\ar[r]^d \ar[d] & \Omega ^4_{\textnormal{cl}}\ar@{=}[d] \\
     \sF_{\Riem} \ar[r]^{\TLC}\ar[d]_{p_1} & \sFBN{\SO_3} \ar[r]^\Gamma &
     \hZ^4\ar[r]^\omega \ar[d]^\mu  & \Omega ^4_{\textnormal{cl}}\\
     Z^4\ar@{=}[rr] && Z^4 } \end{gathered} 
  \end{equation}
The construction of~$\phi $ uses~\eqref{eq:52}: a Riemannian manifold with a
$p_1$-structure produces a 3-form whose differential is the Chern--Weil 4-form
that represents~$p_1$.  

The theory~$\gamma _c$, $c\in \RR$, is constructed by
integrating~$c\phi /24$ and exponentiating.  So if $X$~is a closed oriented
Riemannian 3-manifold with $p_1$-structure, then the partition function is
  \begin{equation}\label{eq:55}
     \gamma _c(X) = \exp\left( \frac{2\pi ic}{24}\int_{X}\phi \right) . 
  \end{equation}

  \begin{remark}[]\label{thm:29}
 As stated in \autoref{thm:28}, one should evaluate field theories in families
of bordisms.  Let us consider the theory~$\gamma _c$ in families.  First, let
$\pi \:X\to S$ be a fiber bundle of smooth manifolds whose fibers are closed
3-manifolds.  Additionally, the relative tangent bundle $T(X/S)\to X$ carries
an orientation and a $p_1$-structure, and there is a relative Riemannian
structure\footnote{This structure gives a Levi--Civita covariant derivative on
the relative tangent bundle $T(X/S)\to X$.} on~$\pi $: that is, an inner
product on $T(X/S)\to X$ together with a horizontal distribution on~$\pi $.
Then 
  \begin{equation}\label{eq:66}
     \gamma _c(X/S) = \exp\left( \frac{2\pi ic}{24}\int_{X/S}\phi \right) 
  \end{equation}
is a smooth function $\gamma _c(X/S)\:S\to \TT\subset \CC^{\times }$.  For $\pi
\:Y\to S$ a smooth fiber bundle with fibers closed 2-manifolds---equipped with
the same geometric structure as above---the imaginary 1-form 
  \begin{equation}\label{eq:67}
     \frac{2\pi ic}{24}\int_{Y/S}\phi 
  \end{equation}
is a connection form on the trivial line bundle over~$S$.
  \end{remark}

  \begin{remark}[]\label{thm:30}
 The composition $\mu \circ \Gamma \circ \TLC$ in~\eqref{eq:54} leads to a
4-dimensional invertible \emph{topological} field theory with ``integer
values''~\cite{F4}, \cite[\S5.4]{FH2}; in the framework of~\S\ref{sec:4} it is
a spectrum map 
  \begin{equation}\label{eq:68}
     \zeta \:\Sigma ^4\MTSO_4\longrightarrow \Sigma ^4H\ZZ. 
  \end{equation}
On a fiber bundle $X\to S$ of closed oriented Riemannian 3-manifolds, the
Chern--Simons theory built from $\Gamma \circ \TLC$ returns the tangential
Chern--Simons invariant $S\to \RR/\ZZ$; the topological theory~$\zeta $
in~\eqref{eq:68} only tracks the homotopy class of this map (or equivalently
the associated $\ZZ$-torsor over~$S$).  On a fiber bundle $Y\to S$ of closed
oriented Riemannian 2-manifolds, the Chern--Simons invariant is a principal
$\RZ$-bundle with connection over~$S$; the topological theory~$\zeta $ returns
the underlying principal $\RZ$-bundle without connection over~$S$ (or
equivalently the associated $\ZZ$-gerbe over~$S$).

The pullback of~\eqref{eq:68} to oriented manifolds with $p_1$-structure is
trivialized, as follows from~\eqref{eq:54}.  Let $\tau $~be the trivialization: 
  \begin{equation}\label{eq:69}
     \begin{gathered} \begin{tikzcd} \Sigma^4\text{MTSO}_4^{p_1} \arrow[r]
     \arrow[rr, bend right=30, "\bone"' pos=0.522, ""{name=B, above,
     pos=0.522}] & |[alias=T]| \Sigma^4\text{MTSO}_4 \arrow[r, "\zeta"] &
     \Sigma^4H\mathbb{Z} \arrow[Rightarrow, from=B, to=T, "\tau" right,
     shorten >=-2pt] \end{tikzcd} \end{gathered}
  \end{equation}
On a closed Riemannian $\wpo$-manifold the trivialization~$\tau $ of the
$\ZZ$-torsor $\zeta (X)$ lifts the Chern--Simons invariant from~$\RR/\ZZ$
to~$\RR$, evident from its expression as the integral~$\int_{X}\phi $.  On a
fiber bundle $Y\to S$ of closed $\wpo$ Riemannian 2-manifolds, the
trivialization~$\tau $ is a (not-necessarily-flat) section of the principal
$\RZ$-bundle with connection over~$S$ that is the Chern--Simons invariant. 
 
There is a similar trivialization of the ``integer-valued'' theory~$[\gamma
_c]$ underlying~$\gamma _c$ for any~$c\in \RR$.
  \end{remark}

   \section{The free spinor field}\label{sec:6}
 
We treat a chiral 2-dimensional spinor field (Majorana--Weyl) in three
Wick-rotated guises.  Before Wick rotation this free theory is defined as a
relativistic quantum mechanical system on 2-dimensional Minkowski spacetime.
The translation group acting on this affine space decomposes under the Lorentz
Spin group~$\Spin_{1,1}$ into a sum of left-~and right-moving translations.
The representation of $\Spin_{1,1}$ that defines the chiral spinor field is the
double cover of the left-moving translations.  The resulting free theory leads
to an irreducible 1-particle representation of the Poincar\'e group.
 
We present the Wick-rotated theory, which is an anomalous theory, in three
incarnations.

  \begin{remark}[]\label{thm:31}
 As pointed out in \autoref{thm:28}, it is important to evaluate the free
spinor field on families of manifolds/bordisms.  We use both holomorphic and
smooth families.  As will be apparent, theories differ depending on the choice
of families, as well as the nature of the output.
  \end{remark}

  \begin{remark}[]\label{thm:33}
 We comment on unitary structures but leave their detailed development to the
future.  
  \end{remark}

  \subsection{Holomorphic presentation}\label{subsec:6.2}

Perhaps the most familiar presentation of the chiral spinor field is as an
anomalous theory over the sheaf~$\sFCww\to \Man_2$, i.e., a theory of
2-dimensional Spin conformal manifolds.  An orientation plus conformal
structure on a 2-manifold~$Y$ is equivalent to a complex structure on~$Y$, and
a compatible Spin structure is then equivalent to a choice of square
root~$\Kh_Y\to Y$ of the canonical bundle $K\mstrut _Y\to Y$.  The Dirac
operator is the $\dbar$-operator coupled to~$\Kh_Y\to Y$:
  \begin{equation}\label{eq:62}
     D_Y=\dbar_Y(\Kh_Y)\:\Omega ^{1/2,\,0}_Y\longrightarrow \Omega
     ^{1/2,\,1}_{Y}, 
  \end{equation}
where $\Omega ^{1/2,\,q}_Y=\Omega ^{0,\,q}_Y(\Kh_Y)$.  This first-order
differential operator is complex skew-adjoint if $Y$~is closed.  Then the
partition function of the free spinor field is the pfaffian of the Dirac
operator, which is an element of the Pfaffian line:
  \begin{equation}\label{eq:63}
     \pfaff D_Y\in \Pfaff D_Y. 
  \end{equation}
This is the top level of an anomalous 2-dimensional field theory~$F$
over $\sFCww\to \Man_2$.

In this version of the theory it is natural to consider \emph{holomorphic}
families.  Thus to a holomorphic fiber bundle $Y\to S$ of closed 1-dimensional
complex manifolds is associated a holomorphic line bundle 
  \begin{equation}\label{eq:70}
     \Pfaff D_{Y/S}\longrightarrow S 
  \end{equation}
together with a holomorphic section $\pfaff D_{Y/S}$.  There is no hermitian
metric or connection on the line bundle~\eqref{eq:70} in this variant.  In
other words, there is no unitary structure on the theory~$F$.

  \begin{remark}[]\label{thm:26}
 The anomaly theory~$\alpha $ comes to us as a once-shifted (or
once-categorified) invertible 2-dimensional theory over the sheaf $\sFCww\to
\Man_2$; the value of~$\alpha $ on a holomorphic family $Y\to S$ is the
holomorphic Pfaffian line bundle~\eqref{eq:70}.  Apparently $\alpha $~does not
extend to an invertible conformal 3-dimensional theory, though we have not
worked out a proof of this assertion.  When we pull back from 2-dimensional
conformal structures to 2-dimensional Riemannian structures, and we replace
holomorphic families by smooth families, then there is an extension, as we
explain in~\S\ref{subsec:6.3}.
  \end{remark}

  \begin{remark}[]\label{thm:35}
 The rational Chern class of the line bundle~\eqref{eq:70} is nonzero for some
families of surfaces.  Introduce a $p_1$-structure: pull back to the sheaf
$\sFCwwp\to \Man_2$.  Then in all holomorphic families of closed surfaces, now
with $p_1$-structure, the rational Chern class does vanish.  Even more is true:
now there is a flat structure on the Pfaffian line bundle.  This line bundle
with flat structure does not depend on conformal structures.  In fact, this
flat Pfaffian line bundle is part of a 3-dimensional topological field theory
over $\sFwwp\to \Man_3$: the theory~$\lambda $ in~\eqref{eq:50}.  In the
remainder of this section we introduce Riemannian metrics to derive this flat
structure.  Not only does it exist for families of closed surfaces, but we will
see~$\lambda $ emerge as a fully local invertible 3-dimensional topological
field theory.
  \end{remark}

  \subsection{Riemannian presentation}\label{subsec:6.3}

As in~\S\ref{subsec:1.1}, let $\sFRww\to \Man_2$ be the sheaf of Riemannian
metrics plus Spin structures.  Consider the pullback of the anomalous
theory~$F$ with anomaly~$\alpha $ along the fiber bundle $\sFRww\to \sFCww$ of
sheaves over~$\Man_2$.  On a Riemannian Spin 2-manifold~$Y$, the Pfaffian line
$\Pfaff D_Y$ carries a Quillen metric.  The pullback of the anomaly
theory~$\alpha $ extends to a \emph{unitary} invertible 3-dimensional
theory~$\ha $ over the sheaf $\sFRww\to \Man_3$: the partition function
of~$\ha $ on a closed Riemannian Spin 3-manifold~$X$ is $\exp(2\pi i\xi
_X/2)$, where $\xi _X$~is the Atiyah--Patodi--Singer $\xi
$-invariant.\footnote{Recall that $\xi _X= \frac 12(\eta _X+\dim\ker D_X)$,
where $\eta _X$~is the Atiyah--Patodi--Singer $\eta $-invariant.  The
$(2,3)$-truncation of~$\ha $ is constructed in~\cite{DF}.}  The
pullback~$\hF$ of the free spinor field theory~$F$ is a \emph{boundary theory}
of~$\ha $ (as theories over $\sFRww\to \Man_3$).

In this Riemannian variant of the theory, it is natural to evaluate on
\emph{smooth} families.  Thus suppose $\pi \:Y\to S$ is a smooth fiber bundle
with a Spin structure on the relative tangent bundle $T(Y/S)\to Y$ and a
relative Riemannian structure on~$\pi $, as in \autoref{thm:29}.  Then $\ha
$~evaluates to the \emph{smooth} Pfaffian line bundle
  \begin{equation}\label{eq:71}
     \Pfaff D_{Y/S}\longrightarrow S 
  \end{equation}
equipped with its Quillen metric and compatible covariant derivative~\cite{F5}.
The Quillen metric and covariant derivative are part of a unitary structure
on~$\ha $.  The free spinor field theory~$\hF$ returns the section $\pfaff
D_{Y/S}$ as in~\eqref{eq:63}.

  \begin{remark}[]\label{thm:32}
 The anomaly theory~$\ha $ can be constructed using differential $KO$-theory.
At least at the top levels, the equivalence between the construction with Dirac
operators and the differential~$KO$ construction can be proved using geometric
variants of Atiyah--Singer index theory.
  \end{remark}

  \subsection{The Witten maneuver and the Adams $e$-invariant}\label{subsec:6.4}

Now we execute the move in~\eqref{eq:32}.  This maneuver is the extension to
invertible field theory of the Atiyah--Patodi--Singer expression
\cite[Theorem~4.14]{APS} for the Adams $e$-invariant.  (The Adams $e$-invariant
for framed manifolds appears in~\S\ref{subsubsec:4.1.1}.)  For this we pull
back along the fibration $\sFRwwp\to \sFRww$ of sheaves over~$\Man_3$ 
.  Then
the (exponentiated) Adams $e$-invariant of~$X$, now a closed 3-dimensional
Riemannian Spin manifold \emph{with a $p_1$-structure}, is
  \begin{equation}\label{eq:64}
     \exp(2\pi i\xi _X/2)\cdot \gamma _{-1/2}(X) = \exp\left( 2\pi i\left[
     \frac{\xi _X}{2} - \frac1{48}\int_{X}\phi  \right] \right) ,
  \end{equation}
where $\gmh$ is defined in~\eqref{eq:55}.  As explained in \autoref{thm:30},
the invertible theory~$\gmh$ carries a nonflat trivialization~$\tau $.  Recall
that a boundary theory of a field theory is a domain wall to the trivial
theory.  Hence a (nonflat) trivialization~$\tau $ of an invertible field theory
is equivalent to a (nonflat) invertible boundary theory~$\hta$.  Then as
theories over $\sFRwwp\to \Man_3$, we can say that $\hF\otimes \hta $ is a
boundary theory of $\ha \otimes \gmh$.  Furthermore, $\ha \otimes \gmh$ is a
\emph{topological} field theory: it factors to a theory over $\sFwwp\to
\Man_3$.  This is precisely the theory~$\lambda $ in~\eqref{eq:50}.

  \begin{remark}[]\label{thm:27}
 If we drop the unitary structure, then the 2-dimensional boundary theory
$\hF\otimes \hta$ depends only on a conformal structure, not a Riemannian
structure, and we identify it with the holomorphic theory~$F$
in~\S\ref{subsec:6.2}, now lifted along $\sFCwwp\to \sFCww$.  Explicitly,
  \begin{equation}\label{eq:65}
     F \bigm/ \sFCwwp\;\textnormal{ is a boundary theory of }\;
     \lambda \bigm/\sFwwp.
  \end{equation}
This achieves the goal set out in \autoref{thm:35}, and it is a standard
picture of the free chiral spinor field in two dimensions: a chiral
2-dimensional conformal field theory as the boundary theory of a 3-dimensional
\emph{topological} field theory.
  \end{remark}

These theories are evaluated on \emph{smooth} families of manifolds/bordisms.
Since $\lambda $~is a topological theory, its values are locally constant.
Thus, for example, to a smooth fiber bundle $Y\to S$ of closed 2-manifolds with
relative $\wwpo$-structure, the theory~$\lambda $ produces a line bundle
over~$S$ with a \emph{flat} covariant derivative.  The boundary theory $F$
returns the pfaffian section $\pfaff D_{Y/S}$.

  \begin{remark}[]\hspace{-5pt}\footnote{We thank Greg Moore for a discussion
about unitary structures.}\label{thm:34}
 The unitary structures in~\S\ref{subsec:6.3} do \emph{not} descend to unitary
structures on $F$ and $\lambda =\ha \otimes \gamma _{-1/2}$.  In fact,
$\lambda $~carries a (rather trivial) unitary structure, being a topological
theory of finite order, but it does not lift to the unitary structure on $\ha
\otimes \gamma _{-1/2}$ given by the Quillen metric and the hermitian structure
on the Chern--Simons line.  (Compare \autoref{thm:15}(4).)  In this vein, we
remark that the \emph{conformal anomaly} obstructs the descent of the unitary
structure on $\ha $ along the fibration $\sFRww\to \sFCww$.
  \end{remark}

 \bigskip\bigskip

\newcommand{\etalchar}[1]{$^{#1}$}
\providecommand{\bysame}{\leavevmode\hbox to3em{\hrulefill}\thinspace}
\providecommand{\MR}{\relax\ifhmode\unskip\space\fi MR }
\providecommand{\MRhref}[2]{%
  \href{http://www.ams.org/mathscinet-getitem?mr=#1}{#2}
}
\providecommand{\href}[2]{#2}

  \end{document}